\documentclass[letter]{article}
\usepackage[margin= 1 in]{geometry}
\usepackage[english]{babel}
\usepackage{amsmath, amssymb, amsthm}
\usepackage{graphicx}
\usepackage[colorinlistoftodos]{todonotes}
\usepackage{multicol}

\title{Implementing $\Diamond P$ with Bounded Messages on a Network of ADD Channels}

\author{Saptaparni Kumar, Jennifer Welch. \\ {\em{\small Texas A\&M University, College Station, Texas- 77840}}}

\usepackage{url}      
\usepackage{comment}
\usepackage{algorithm}
\usepackage{algcompatible}

\newcommand{\myfontsize}{\fontsize{8}{9}\selectfont}

\algblockdefx{PRE}{ENDPRE}%
  {\textbf{precondition:}}%
  {\textbf{end precondition}}

\algblockdefx{EFF}{ENDEFF}%
  {\textbf{effect:}}%
  {\textbf{end effect}}
  
\makeatletter
\ifthenelse{\equal{\ALG@noend}{t}}%
  {\algtext*{ENDPRE}}
  {}%
\makeatother

\makeatletter
\ifthenelse{\equal{\ALG@noend}{t}}%
  {\algtext*{ENDEFF}}
  {}%
\makeatother

\newtheorem{observation}{Observation}
\newtheorem{definition}{Definition}
\newtheorem{theorem}{Theorem}
\newtheorem{lemma}{Lemma}

\usepackage[pdftex]{hyperref}
\hypersetup{
pdfauthor={Saptaparni},
colorlinks,
citecolor=black,
filecolor=black,
linkcolor=black,
urlcolor=black,
breaklinks=true,
}

\begin{document}
\maketitle

\begin{abstract}
We present an implementation of the eventually perfect failure detector ($\Diamond P$) from the original hierarchy of the Chandra-Toueg~\cite{Chandra96} oracles on an arbitrary partitionable  network composed of unreliable channels that can lose and reorder messages. Prior implementations of $\Diamond P$ have assumed different partially synchronous models ranging from bounded point-to-point message delay and reliable communication to unbounded message size and known network topologies. We implement $\Diamond P$ under very weak assumptions on an arbitrary, partitionable network 
composed of  {\em Average Delayed/Dropped (ADD)} channels~\cite{Sastry} to model unreliable communication. 
Unlike older implementations, our failure detection algorithm uses bounded-sized messages to eventually detect all nodes that are unreachable (crashed or disconnected) from it.
\end{abstract}

\section{Introduction}

{\em Failure detectors} were proposed by Chandra and Toueg~\cite{Chandra96} as oracles to be used to identify failed nodes in a crash-prone asynchronous message-passing system. Each node has  a failure detector module that can be queried for information about which nodes in the system have crashed. 
Unreliable failure detectors can give wrong information by incorrectly suspecting correct nodes, and/or not suspecting crashed nodes. In spite of that, many oracles are powerful enough to solve important distributed problems that are otherwise  unsolvable in asynchronous systems with even one crash failure. The hierarchy of the Chandra-Toueg oracles~\cite{Chandra96} was originally introduced to circumvent  the FLP~\cite{Fischer85} impossibility result for solving consensus in a crash-prone asynchronous message-passing system,  by identifying crashed nodes and distinguishing them from slow nodes. These failure detectors can be described by their accuracy and completeness properties.
Their implementation in practice requires some degree of partial or even full synchrony.  Paper~\cite{Freiling11} provides a very informative survey on the failure detector abstraction both as building blocks for the design of reliable distributed algorithms and as computability benchmarks.

There are two main lines of research in the area of failure detectors. The first one involves implementing failure detectors on increasingly weaker  system models that represent practical applications and the second one involves finding the weakest failure detector for solving a given problem. We contribute to the first line of research by presenting a novel implementation of an {\em eventually perfect} failure detector ($\Diamond P$) from the original Chandra-Toueg hierarchy 
on an arbitrary partitionable network of unknown topology composed of  ill-behaved channels.  

 In \cite{Aguilera99}, the definitions of accuracy and completeness for failure detectors were extended to partitionable networks. $\Diamond P$ for partitionable networks satisfies {\em strong completeness} and {\em eventual strong accuracy}. Intuitively, $\Diamond P$ can give incorrect information about the nodes in the system for an unknown finite amount of time, after which, it provides perfect information about all nodes in the system. Strong completeness is satisfied if the failure detector of each node eventually suspects all nodes that are unreachable (i.e., crashed or disconnected) from it. Eventual strong accuracy is satisfied if the failure detector of every node  eventually stops suspecting all nodes that are reachable from it. 
The paper by Chen et al.~\cite{Chen02} studies  accuracy and completeness properties (quality of service) of failure detectors and quantifies how fast different implementations of oracles detect failures and how well they avoid false suspicions. Papers~\cite{Renesse98} and~\cite{Hutle} have discussed algorithms to perform failure detection on arbitrary networks composed of unreliable channels using counters as heartbeats for the nodes in the system. Unlike our algorithm, the message sizes in their algorithms are unbounded. Papers~\cite{Fraigniaud16} and \cite{Fernández-Campusano14} do failure detection using bounded sized messages, but unlike our work, they assume that the underlying communication channels are reliable. Our contribution in this paper is a little of both. We present a novel algorithm that implements $\Diamond P$ in an \textbf{arbitrary (partitionable)} network composed of channels that provide very weak guarantees (\textbf{unreliable channels}), using messages that are \textbf{bounded in size}.

The framework of ADD (Average Delayed/Dropped) channels was introduced by Sastry and Pike~\cite{Sastry} who exploited the channel properties to implement $\Diamond P$ on cliques; the properties of an ADD channel are valid only for one-hop networks. The motivation for this paper was based on extending the failure detector from a fully connected network of ADD channels to  any arbitrary partitionable network composed of ADD channels using bounded size messages. 
The ADD channels are a realistic partially synchronous model of ill-behaved channels that can lose and reorder messages. 

Our failure detection algorithm uses bounded sized heartbeats, timeouts and path information to determine if there is a correct  path (all nodes on this path are correct) between two nodes. Periodically, every node sends out its own heartbeat to its neighbors. Every node has an estimated timeout value for its neighbors and if a node does not hear from its neighbor within this estimated time, it suspects the neighbor to be crashed. The timeout value gets incremented every time a neighbor is incorrectly suspected. For a node that is not a neighbor, the algorithm goes through a list of paths from itself to the node and evaluates if the node is reachable via at least one of these paths (i.e., all nodes on a path are correct according to the failure detector). If no such path exists, it suspects this node.

\section{Model and Definitions}

\subsection{Communication Model}
We are considering the original definition~\cite{Sastry} of the Average Delayed/Dropped (ADD) channel in which every link connecting two nodes in the network is composed of two unidirectional ADD channels, one in each direction.
All messages sent on an ADD channel are eventually delivered or lost but not duplicated.
The messages can be partitioned logically into two disjoint sets: privileged and unprivileged. Unprivileged messages have no timing or reliability guarantees and may be arbitrarily delayed or even dropped. However, an ADD channel provides the following guarantee for privileged messages.
For every ADD channel there is an \textbf{unknown} upper bound $d$ on the delay of all privileged messages and a \textbf{unknown} upper bound $r$ on the number of unprivileged messages sent between any consecutive pair of privileged messages. Intuitively this means that out of $r$ consecutive messages sent on an ADD channel at least one message is guaranteed to be delivered within $d$ time.
\subsection{Network Model}

The distributed system consists of a set $\Pi$ of $n$ nodes connected in an arbitrary topology by ADD channels. Nodes may fail only by crashing. Nodes that never crash are called $correct$ nodes and those that have not crashed yet are called $live$ nodes. Each node that $crash$es remains crashed forever. 
Each node knows who its neighbors are. Nodes also know the names (ids) of all the nodes in the system. This assumption is not trivial, as Jimenez et al.~\cite{Jimenez06} show that without this assumption, no failure detector can be implemented, even in a fully synchronous system with reliable links. 

Nodes communicate with neighbors only via local  point-to-point communication. Each node has a local clock which generates ticks at a constant rate. Different clocks can tick at different rates and can be unsynchronized. The network is initially a connected graph but may eventually be partitioned as nodes start crashing. We call this network a {\em partitionable network}.
\begin{definition}
The \textbf{network graph} at time $t$ is a subgraph of the initial graph obtained by deleting all nodes (and their incident links) that are crashed at time $t$.
\end{definition} We denote the network graph at time $t$ as $G(t)$. 


\section{Problem Statement}

We address the problem of implementing, with bounded-sized messages, an {\em Eventually Perfect} ($\Diamond P$) failure detector that satisfies the following on an arbitrary partitionable network $G$ composed of ADD channels. 
For each node $p$, there is a function from $p$'s state to the set of nodes that $p$ suspects.
 In every execution there exists a time $t_f$ such that for every $t > t_f$ and every correct node $p$,
 \begin{itemize}
 \item {\em Strong Completeness:} for every node $q$ that is disconnected from $p$ in $G(t)$, $p$ suspects $q$ at time $t$
 \item {\em Eventual strong accuracy:} for every node $q$ that is connected to $p$ in $G(t)$, $p$ does not suspect $q$ at time $t$.
 \end{itemize}

\section{$\Diamond P$ Algorithm}
\begin{algorithm}
\caption{{\sc $\Diamond$P:} Eventually Perfect Failure Detector, Code for node p}
\label{algo:Heartbeat}
\begin{multicols*}{2}
\begin{algorithmic}[1]

\myfontsize
\item[]{\bf Constants:}
\STATE $neighbors$ \COMMENT{list of neighbors of $p$.}
\STATE $T$ \COMMENT{integer;  time between successive heartbeats}

\item[] {\bf Variables:}
\STATE $clock()$ \COMMENT{local clock}
\STATE $last\_contact[.]$ \COMMENT{array of clock times for all neighbors, last time $p$ received a message about that node; initially $last\_contact[q] = 0$, for all $q\in neighbors$}
\STATE $suspect\_local[\cdot]$ \COMMENT{array of booleans for all  nodes; initially $suspect\_local[q]$ = false, for all $q\in \Pi$. This stores the failure information for nodes in $p$'s connected component}
\STATE $paths[\cdot]$ \COMMENT{Array of sets of paths or sequences of node ids. $paths[q]$ is the set of  paths taken by the heartbeat messages from $q$ to $p$; Initially $paths[p] = \{p\}$, $paths[q] =\{q\cdot p\}$  for $q\in neighbors$ and $paths[r] = \O $ for all others}
\STATE $suspect[\cdot]$ \COMMENT{array of booleans for all nodes; it is true for all nodes suspected to have failed; initially $suspect[q]$ = false, for all $q\in \Pi$.  This stores the failure information stored in $suspect\_local[\cdot]$ and also information derived from the $paths[\cdot]$ variable}

\STATE $timeout[\cdot]$ \COMMENT{Array of time-outs for all neighbors. $timeout[q]$ is the estimated maximum time between the  receipt of successive messages about neighbor $q$; initially $timeout[q]=T$, for all $q\in neighbors$}
\item[]
\hrulefill
\item[]

\STATE $\langle$ {\bf output} \textit{Send heartbeat to neighbors} $\rangle$:
\PRE
	\STATE $clock() = n\cdot T $ for $n \in \mathbb{N}$     \COMMENT{clock() is an integer multiple of $T$}\label{line:boundedPersistent}
\ENDPRE
\EFF
	\STATE   $send \langle suspect\_local[\cdot],paths[\cdot],p \rangle$ to every neighbor of $p$ \COMMENT{Send heartbeat with a copy of the local suspect list and the path list} \label{line:broadcast}
\ENDEFF

\item[]

\item[]

\STATE $\langle$ {\bf input} $recv \langle suspect\_rcv[\cdot],path\_sets[\cdot],q \rangle $ $\rangle$: 
\item[]
\EFF

	\IF[neighbor wrongfully suspected]{$suspect\_local[q] = true$}
    		\STATE $timeout[q] := 2 \cdot (clock()-last\_contact[q])$
    		\STATE $suspect\_local[q] := false$\COMMENT{Stop suspecting this neighbor}\label{line:stopSuspectingNeighbor}            
    	\ENDIF
    \STATE $last\_contact[q] := clock()$ 
   	
    	\item[]
    	\FOR{all $r\notin neighbors$}
    	 \STATE $hop\_from\_msg$ := Length of the shortest path in $path\_sets[r]$ not containing $p$ or a node $u$, $u\ne r$ with $suspect\_rcv[u] = true$ \label{line:hopEstimate}
    	 \STATE $hop$ := Length of the shortest path in $paths[r]$ not containing a node $u$, $u\ne r$ with $suspect\_local[u] = true$ \label{line:hopSelfEstimate}
    	 
         \IF{$hop\_from\_msg<hop $} \label{line:hopCheck}
        		 \STATE $suspect\_local[r] := suspect\_rcv[r] $\label{line:suspectUpdate}
         \ENDIF
         \FOR{each path, $\pi \in path\_sets[r]$ that does not contain $p$}\label{line:eliminateCycles}
           		\STATE $paths[r] := paths[r] \cup \{\pi\cdot p\} $ \COMMENT{Append new paths to the $paths[r]$ set}\label{line:addAllPaths}
         \ENDFOR
     \ENDFOR
    \item[]
 \item[]

	\STATE $suspect[\cdot] := suspect\_local[\cdot]$\label{line:derivedStart}
    \STATE Let $Sus$ be the set of all $u$ with $suspect\_local[u] = true$
    \FOR {all $r \notin neighbors$}
		 \IF{all paths in $paths[r]$ contain a $u\in Sus$ and $u\ne r$ }\label{line:allPathsHaveSus}
		    \STATE $suspect[r] := true$\label{line:deivedTrue}
		\ENDIF
	\ENDFOR\label{line:derivedEnd}
	 	\item[]
\ENDEFF	
\item[]		
\item[]

\STATE $\langle$ {\bf output} $timer\_expiry(q)$ $\rangle$: \COMMENT{The timer for a neighbor expires} \label{onlyNeighbors}
\PRE
	\STATE $timeout[q] = clock() - last\_contact[q]$
\ENDPRE
\EFF
	\STATE $suspect\_local[q] := true$  \label{endOnlyNeighbors}
\ENDEFF

\item[]
\end{algorithmic}
\end{multicols*}
\end{algorithm}


Algorithm~\ref{algo:Heartbeat} implements  $\Diamond P$ over the partitionable network of ADD channels. Every node $p$ maintains a variable $suspect\_local[\cdot]$ which is an array of booleans to store information about nodes in $p$'s connected component and a variable $paths[\cdot]$ which is an array of paths (sequences of node ids). For example, the  variable $paths[q]$ stores  the  paths between $p$ and $q$ that $p$ has learned about from messages from its neighbors. Node $p$ also maintains an array of booleans $suspect[\cdot]$ which stores failure information from $suspect\_local[\cdot]$ along with extra information derived from  $paths[\cdot]$. Its value is set to true for nodes estimated to have crashed (using failure information from $suspect\_local[\cdot]$)  or nodes that are estimated to be disconnected from $p$ (using information from  $paths[\cdot]$).  If $p$'s $suspect[q]$ variable is $true$, then we say $p$ {\em suspects} $q$.

Node $p$ sends out a heartbeat message  containing the variables  $suspect\_local[\cdot]$ and $paths[\cdot]$  to its neighbors every $T$ units of time. $T$ may be chosen differently by every node.  The smaller the value of $T$, the faster the failure detection algorithm will converge, but if $T$ is too small, the network may be over crowded with heartbeat messages.  
When $p$ receives a heartbeat message from a neighbor $q$, it records its current clock time ($clock()$) as the $last\_contact$ value.  If $q$ was wrongly suspected  as the $timeout[q]$ value was estimated to be too small, $p$ stops suspecting $q$ by setting $suspect\_local[q]$ to $false$ and increments its $timeout[q]$ value for $q$. Then, $p$ extracts information about the rest of the network from the message from $q$. For all nodes $r$ that are not neighbors of $p$, $p$ calculates $q$'s distance from $r$ and compares it to its own calculated distance from $r$. If $q$ is calculated to be nearer to $r$ than $p$, then $p$ adopts $q$'s information about $r$. Node $p$ goes through all the paths in the variable $path\_sets[\cdot]$ received in the message from $q$ and sees if $p$ is already included in those paths. If $p$ learns about any  path $\pi$, from $r$ to $q$ that does not include $p$, it adds the path $\pi\cdot p$ to its $paths[r]$ set.

 Then, $p$ updates its $suspect[\cdot]$ variables using information from the $paths[\cdot]$ variable about nodes that are not in $p$'s connected component. As there are no paths from these nodes to $p$,  information about these nodes is not received directly. Node $p$ checks if at least one path in the $paths[r]$ set has all nodes that have $suspect\_local[\cdot]$ set to  $false$. If not, it suspects $r$ (i.e., $suspect[r]$ is set to $ true$ ).

\section{Proof of Correctness}

To prove correctness, we need to show that the implementation in Algorithm~\ref{algo:Heartbeat}  satisfies
 \textbf{strong completeness} and \textbf{eventual strong accuracy}.
 We describe some  lemmas to prove that Algorithm~\ref{algo:Heartbeat} implements  $\Diamond P$.  
  
 Lemma~\ref{lemma:interArrivalTime} shows that there is an upper bound on the inter-arrival time of heartbeats at all correct nodes from correct neighbors. Lemma~\ref{lemma:timeoutsStopChanging} shows that eventually, the time-out estimates for neighbors stop changing. Lemmas~\ref{lemma:neighborsSuspect} and~\ref{lemma:neighborsNeverSuspect} show that eventually all correct nodes suspect crashed neighbors and stop suspecting correct neighbors. Lemma~\ref{lemma:pathsIsCorrect} shows that eventually the $paths[q]$ variable at a correct node $p$ contains all the paths between $p$ and $q$ in the final network graph. Lemma~\ref{lemma:eventualCorrectnessSamePartition}, along with Theorem~\ref{corollary:eventualStrongAccuracySamePartition}, proves eventual strong accuracy. Theorem~\ref{lemma:strongCompletenessDifferentPartition} proves strong completeness. Finally, Theorem~\ref{theorem:DiamondP} shows that Algorithm~\ref{algo:Heartbeat} implements  $\Diamond P$ using Theorems~\ref{corollary:eventualStrongAccuracySamePartition} and~\ref{lemma:strongCompletenessDifferentPartition}.  
 
 We use a subscript to denote which node a variable belongs to; for example $p$'s $suspect\_local[q]$ variable will be denoted as $suspect\_local_p[q]$. From here on we refer to nodes that are neighbors with respect to the initial network graph as {\em initial neighbors}. 

\begin{lemma} \label{lemma:interArrivalTime}
There exists an upper bound on the inter-arrival time of heartbeats for correct initial neighbors.
\end{lemma}

\begin{proof}
The properties of an ADD channel guarantee that at least one in every $r$ messages sent on a channel is privileged and thus is delivered within {\color{black}$d$}
 time. Thus the maximum time between the receipt of two consecutive heartbeat messages sent on Line~\ref{line:broadcast}  of Algorithm~\ref{algo:Heartbeat} at any neighbor $q$ of $p$ is {\color{black}$(r+1)\cdot T + d $} where $r$ and $d$ are the ADD channel parameters and $T$ is a constant in $p$'s algorithm.
 
 \end{proof}



Let $t^*$ be the time when all the failures have occurred. 
Let $t^{**} \geq t^{*}$ be the time when all messages (privileged and unprivileged) from all crashed nodes have been delivered or lost. We call the network graph after $t^*$ the {\em final network graph} and denote it by $\mathbb{G}$.

\begin{lemma} \label{lemma:timeoutsStopChanging}
 Eventually, all timeout[q] variables stop changing. 
\end{lemma}
\begin{proof}
Let $q$ be an initial neighbor of $p$. 
Assume by contradiction that the $timeout_p[q]$ variable at node $p$ 
keeps changing forever. 

If $q$ is a correct node, 
this would mean that the inter-arrival times of messages from $q$  to $p$ keeps increasing and so does the $timeout_p[q]$ value. But by Lemma~\ref{lemma:interArrivalTime}, there is an upper bound on this value. Thus our assumption that $timeout_p[q]$  changes forever was false.

If $q$ is a crashed node,  
the number of messages sent to $p$ by $q$ is finite and they are eventually delivered to $p$ or lost. The $timeout_p[q]$ value may keep on increasing till $t^{**}$. After this, $p$ never receives a message from $q$ and the thus $timeout_p[q]$  stops changing after this. Thus our assumption that $timeout_p[q]$  changes forever was false.
\end{proof}

\begin{observation}\label{obs:suspectIsDerivedSuspect}
For a correct node $p$, the $suspect_p[q]$ variable for an 
initial neighbor 
$q$ is equal to  $suspect\_local_p[q]$ at all times. 
\end{observation}

\begin{lemma}\label{lemma:neighborsSuspect}
There exists some time $t$ after which all correct initial neighbors 
 $p$  of  a crashed node $q$ in $\mathbb{G}$ 
suspect $q$. 
\end{lemma}
\begin{proof}
Let  us assume that $q$ crashes at time $t_c $. 
From lines ~\ref{line:stopSuspectingNeighbor} and ~\ref{onlyNeighbors} - \ref{endOnlyNeighbors} of Algorithm ~\ref{algo:Heartbeat} we know that the $suspect\_local[\cdot] $ variables for initial neighbors 
are set only by the nodes themselves (i.e., nodes do not update information about their initial neighbors from other nodes). From Lemma~\ref{lemma:timeoutsStopChanging}, we know that by some time  $t^f$, the variable $timeout_p[q]$ stops changing. 
So, by time $t = (\max{(t^{f} ,t^{**})} + timeout_p[q]) $, $p$ sets $suspect\_local_p[q]$ to $true$ permanently. By Observation~\ref{obs:suspectIsDerivedSuspect}, we know that $suspect_p[q] = suspect\_local_p[q]$. Thus after $t$, $p$ suspects $q$. 
\end{proof}

\begin{lemma}\label{lemma:neighborsNeverSuspect}
There exists some  time $t$ after which all correct initial neighbors $p$  of  a correct node $q$  stop suspecting $q$. 
\end{lemma}
\begin{proof}
Lemma~\ref{lemma:interArrivalTime} states that there is an upper bound on the inter-arrival time of messages from $q$ to $p$. 
 Lemma~\ref{lemma:timeoutsStopChanging} states that $timeout_p[q]$ eventually stops changing. Let the time at which $timeout_p[q]$ stops changing be $t$ and the final value of $timeout_p[q]$ be $\tau$. Thus after $t$, $p$ receives a message from $q$ within every $\tau$ time and thus, $p$ never sets $suspect\_local_p[q]$ to $true$. By Observation~\ref{obs:suspectIsDerivedSuspect}, we know that $suspect_p[q] = suspect\_local_p[q]$, thus after $t$, $p$ stops suspecting $q$. 

\end{proof}

\begin{lemma}\label{lemma:pathsIsCorrect}
For all $p$ and $q$ such that $q$ is not an initial neighbor of $p$, 
eventually the  $paths_p[q]$ variable contains all the paths in $\mathbb{G}$ 
from $q$ to $p$ and $paths_p[q]$ stops changing.
\end{lemma}
\begin{proof}
Let us assume that there exists a path $\pi$ in $\mathbb{G}$ 
 that is never included in $paths_p[q]$ variable at $p$. Let the length of this path be $k$ and the nodes in this path be $q\cdot q_1 \cdot  q_2 \cdot \cdots \cdot q_{k-1} \cdot p$. From Lemma~\ref{lemma:interArrivalTime}, we know that each node in this path receives a message from the previous node in this path infinitely often as all nodes in this path are correct. From line number 
\ref{line:addAllPaths}, we know that each of these paths is appended to the $paths_p[q_i]$ variable of each node $q_i\notin neighbors_p$. So, when $p$ gets a message from $q_{k-1}$ with the path $q\cdot q_1 \cdot  q_2 \cdot \cdots \cdot q_{k-1}$ in it, it adds the path $q\cdot q_1 \cdot  q_2 \cdot \cdots \cdot q_{k-1}\cdot p$ to its $paths_p[q]$ variable. 

Note that all $paths_p[q]$ end with $p$. After $paths_p[q]$ contains all the paths between $q$ and $p$ in  $\mathbb{G}$, if $p$ learns about a new path from $q$ to $p$, it must be a cycle with $p$ in it already. Line number~\ref{line:eliminateCycles} checks if $p$ is in the path already and if this is true, $p$  ignores it. Thus the value of $paths_p[q]$ stops changing once it learns about all paths between $q$ and $p$ in $\mathbb{G}$. 

\end{proof} 

A node $p$ has {\em perfect information} about node $q$ at time $t >t^*$, if any one of the following hold:
\begin{itemize}
 \item If $q$ is 
  in $p$'s connected component in $\mathbb{G}$, $suspect_p[q] = false$ at $t$.
 \item If $q$ is crashed and is an initial neighbor  
  of $p$'s connected component $C$ in $\mathbb{G}$, 
 $suspect\_local_p[q] = true$ at $t$ (This is because $suspect\_local_p$ stores information about $C$ and its crashed initial neighbors). Note that if $suspect\_local_p[q]$ is  $true$, then line numbers~\ref{line:derivedStart} to \ref{line:derivedEnd} imply that $suspect_p[q]$ is set to $true$ as well.
  \item If $q$ is not in $p$'s connected component in $\mathbb{G}$, $suspect_p[q] = true$ at $t$. 
 \end{itemize}
 

\begin{lemma}\label{lemma:eventualCorrectnessSamePartition}
For all $k$, there exists a time $t_k$ such that, for all $t \ge t_k$, all correct nodes $p$ in a connected component $C$ in $\mathbb{G}$, have perfect information about all  nodes $q$ that are either correct and in $C$ or are crashed and are initial neighbors 
 of $C$, 
at a distance at most $k$ from $p$ in $\mathbb{G}$. 
\end{lemma}
\begin{proof}
 We prove this lemma by strong induction on the distance  $k$ of node $p$ from $q$ in $\mathbb{G}$. 
\textit{Base case: $k =1$.} 
From Lemmas~\ref{lemma:neighborsSuspect} and \ref{lemma:neighborsNeverSuspect}, we know that there exists a time when all correct initial neighbors 
have perfect information about $q$. 
\\\textit{Inductive hypothesis:} 
Let us assume that all nodes that are at most $ k-1 $ hops away from $q$  in $\mathbb{G}$ have perfect information about $q$ after some time $t_{k-1}$. \\
\textit{Inductive step:} Let $p$ be a node at a distance $k$  from $q$ in $\mathbb{G}$. If $q$ is correct,  there exists a path $q\cdot \pi \cdot p$ from $q$ to $p$ (since $p,q\in C$). If $q$ is crashed, there exists a path $\pi\cdot p$ from an initial neighbor of $q$ to $p$ (since $q$ is an initial neighbor of $C$).  Let $r$ be a node on this path $k-1$ hops away from $q$ (note, $r$ is an initial neighbor of $p$). From the inductive hypothesis we know that $r$ has perfect information about $q$ after time $t_{k-1}$. 
Before time $t_{k-1}$, $r$ sent only a finite number of messages to $p$. Let $t^{\dagger}$ be the time when all  messages  sent before $t_{k-1}$ are either delivered to $p$  or lost. All messages that $r$ sends to $p$ after time $t_{k-1}$  have perfect information about $q$ in them. 

From Lemma~\ref{lemma:interArrivalTime}, we know that there is an upper bound on the time between two consecutive receive events from $r$ to $p$. Let this upper bound be $\tau$. By Lemma~\ref{lemma:pathsIsCorrect}, after some time $t_{a}$, all paths from $q$ to $r$ in $\mathbb{G}$ 
are appended to $paths_r[q]$ and sent to $p$ in the variable $path\_sets[q]$. So, after time $t' = (\max\{t_{k-1},t^{\dagger}, t_{a}\}+2\tau)$, all messages $p$ gets  from $r$ have perfect information about $q$ and contain all paths between $q$ and $r$ (in $\mathbb{G}$). 

When $p$ processes a message from $r$ after $t'$ we argue that the value of $hop\_from\_msg$ for $q$ on line number~\ref{line:hopEstimate} is $k-1$. 
 $hop\_from\_msg$ cannot be greater than $ k-1$ because 
 the variable $path\_sets[q]$  contains $q\cdot \pi$ which is of length $k-1$.  Also, $hop\_from\_msg$ cannot be less than $k-1$ because $r$ has perfect information about all nodes at a distance $k-1$ from $r$, and so, the estimate for $hop\_from\_msg$ will discard all paths with length less than $ k-1$ as they are no more available in $\mathbb{G}$. 
We also argue that the value of $hop$  for $q$ on line number~\ref{line:hopSelfEstimate} is greater than $ k-1$. This is because, by the inductive hypothesis, $p$ has perfect information about nodes that are $k-1$ hops away from $p$ in $\mathbb{G}$. So, all entries in $paths_p[r]$ with length at most $ k-1$ are discarded as they are correctly estimated to have a at least one crashed node in them. Thus, the value of $hop$  for $q$ is greater than $k-1$. 
 As a result, the `if' condition on line number~\ref{line:hopCheck} is satisfied and $p$ adopts $r$'s information about $q$. By the  inductive hypothesis, this information is perfect (note that $p$ adopts only $r$'s $suspect\_local_r[q]$ variable which currently contains perfect information about $q$). 

We still have to show that $p$'s  $suspect_p[q]$ variable does not get set to something different (from the one set by the message from $r$) 
by a message coming from a node with wrong information about $q$. Let us assume that by contradiction, $p$ gets a message from a node $s$ that is at a distance $i > k$ from $q$ and the  value of $hop\_from\_msg$ for $q$ on line number~\ref{line:hopEstimate} is miscalculated to be  at most $ k-1$. This scenario is possible only if there is a  path $\pi'$ in $paths_s[q]$ with $|\pi'| \leq k-1$ that is wrongly assumed to exist in $\mathbb{G}$. 
However, since $s$ is at a distance $i $ greater than $ k$ from $q$  in $\mathbb{G}$, $\pi'$ must have a crashed node in it. Let $z\in \pi'$ be the crashed node that $s$ has wrong information about. 
Since $z$ is less than $k$ hops away from $s$, by the induction hypothesis, $s$ already has perfect information about $z$. Thus the assumption that the value of $hop\_from\_msg$ for $q$ on line number~\ref{line:hopEstimate} is calculated to be at most $k-1$ is incorrect and $p$ permanently possesses perfect information about $q$.


\end{proof}
Theorem~\ref{corollary:eventualStrongAccuracySamePartition} proves eventual strong accuracy and Theorem~\ref{lemma:strongCompletenessDifferentPartition} proves strong completeness.

\begin{theorem}\label{corollary:eventualStrongAccuracySamePartition}
Eventually, two correct nodes $p$ and $q$ in the same connected component of $\mathbb{G}$  stop suspecting each other.
\end{theorem}
\begin{proof}
The proof is direct from Lemma~\ref{lemma:eventualCorrectnessSamePartition}.
\end{proof}

\begin{observation}\label{observation:pathPartition}
Let $C$ be a connected component in $\mathbb{G}$. 
Let $q$ be a node in $ \Pi-C$. For every path $\pi$ between $p\in C$ and $q$ in the original network graph, there exists a node $r$  such that  $r$ is a crashed initial neighbor 
 of  $C$.
\end{observation}

\begin{theorem}\label{lemma:strongCompletenessDifferentPartition}
Every correct node eventually suspects all nodes that are not in its connected component in $\mathbb{G}$.
\end{theorem}
\begin{proof}
Let $C$ be a connected component in $\mathbb{G}$. We show that all $p\in C$ eventually suspect all $q \in \Pi - C$. Since originally the network was a connected graph, there was a path from all $q$ to $p$ in the initial network graph. We separate this proof into two parts:
\begin{itemize}
\item $q$ is an initial neighbor of $p$. 
In this case, the proof is direct from Lemma~\ref{lemma:neighborsSuspect}.

\item $q$ is not an  initial neighbor of $p$. 
Let $\pi$ be a path in $paths_p[q]$. From Observation~\ref{observation:pathPartition}, we know that all paths from $q$ to $p$ have a crashed node $r$ that is an initial neighbor of a node in $C$.  From 
 Lemma~\ref{lemma:eventualCorrectnessSamePartition}, we know that after some time $t$, $p$ has perfect information about  $r$. 
Thus, all $\pi \in paths_p[q]$ have a node $r$ that has $suspect\_local_p[r] = true$. When $p$ calculates the $suspect$ variable for $q$ on line numbers~\ref{line:derivedStart} to \ref{line:derivedEnd}, the if condition on line number~\ref{line:allPathsHaveSus} is satisfied and $suspect_p[q]$ is set to $true$ on line number~\ref{line:deivedTrue}.

Now we show that this value of the $suspect$ variable is not reversed. Since from Lemma~\ref{lemma:pathsIsCorrect}, we know that $paths_p[q]$ stops changing and the information about all nodes $r$ is never reversed, we can safely conclude that line number~\ref{line:allPathsHaveSus} is always satisfied henceforth and  $suspect_p[q]$ always remains $true$.
\end{itemize}

\end{proof}

\begin{theorem}\label{theorem:DiamondP}
Algorithm~\ref{algo:Heartbeat} implements an eventually perfect failure detector for partitionable networks  using bounded size messages.
\end{theorem}
\begin{proof}
This proof is direct from Theorems~\ref{corollary:eventualStrongAccuracySamePartition} and ~\ref{lemma:strongCompletenessDifferentPartition} which prove eventual strong accuracy and  strong completeness respectively.

The messages sent by a node $p\in \Pi$ have the variables $suspect\_local[\cdot]$ and $paths[\cdot]$ in them. Note that both these variables are bounded in size. The $suspect\_local[\cdot]$ variable has $n$ booleans and so has size $ n$ bits.  The $paths[\cdot]$ variable contains only simple paths between nodes which can be at most $O(n \cdot n!)$ bits (in the case of a complete graph) and so has message size at most $O((n +1)!)$ bits. Thus, the messages used in this algorithm are bounded in size. 
\end{proof}


\section{Conclusion}
We have implemented the eventually perfect failure detector ($\Diamond P$) in a weak, arbitrary, partitionable network model composed of unreliable, partially synchronous ADD channel with unbounded message loss and unbounded message delay for a majority of the messages. This work is an important step towards understanding the minimal assumptions on network topology, message sizes, reliability of channels and partial synchrony necessary to implement this oracle. The algorithm is quite practical for sparsely connected graphs as the number of paths between two nodes (and the message size) will be $\ll$ $(n+1)!$. Even though the message size for this algorithm is bounded, can we do better than our current results using smaller messages or fewer messages? We think that these are important questions that need to be answered in the future.

\end{document}